\documentclass[submission,copyright,creativecommons]{eptcs}
\providecommand{\event}{SOS 2007} % Name of the event you are submitting to
\usepackage{breakurl}             % Not needed if you use pdflatex only.
\RequirePackage{ifthen}
\usepackage{makeidx}
\usepackage[usenames,dvipsnames]{color}
\usepackage[switch, modulo]{lineno}
\usepackage{amsmath}
\usepackage{amssymb}
\usepackage{amsthm}

\iftrue
\renewcommand\refname{REFERENCES}
\usepackage{titlesec}
\titleformat{\section}[runin]
  {\bfseries}
  {\thesection}{1em}{}[.]
\titlespacing{\section}
  {0pt}
  {1em}
  {1em}
\fi

\newtheorem{DEF}{Definition}{\bfseries}{\rmfamily}
{\bfseries}{\rmfamily}
{\bfseries}{\itshape}
{\bfseries}{\rmfamily}
{\bfseries}{\rmfamily}
\newtheorem{LEM}{Lemma}{\bfseries}{\itshape}
\newtheorem{THM}{Theorem}{\bfseries}{\itshape}
{\bfseries}{\itshape}
{\bfseries}{\rmfamily}

\newcommand{\comment}[1]{{}}
\newcommand{\COMMENT}[1]{}
\newcommand{\TZ}[1]{{#1}}
\newcommand{\MD}[1]{{#1}}

\newcommand{\ASET}[1]{\{#1\}}
\newcommand{\PAR}{\ | \ }

\newcommand{\protocolS}{{\textbf{protocol}}}
\newcommand{\roleS}{{\textbf{role}}}

\newcommand{\fromS}{{\textbf{from}}}
\newcommand{\toS}{{\textbf{to}}}
\newcommand{\choiceS}{{\textbf{choice}}}
\newcommand{\atS}{{\textbf{at}}}

\newcommand{\orS}{{\textbf{or}}}
\newcommand{\runS}{\textbf{run}}

\newcommand{\natk}{\kw{nat}}
\newcommand{\boolk}{\kw{bool}}
\newcommand{\intk}{\keyword{int}}
\newcommand{\stringk}{\keyword{str}}
\newcommand{\unitk}{\keyword{unit}}

\newcommand{\Context}{\mathtt{C}}
\newcommand{\GAM}{\textit{G}}

\newcommand{\sort}{\mathit{S}}
\newcommand{\kend}{\mathtt{end}}
\newcommand{\M}{\textit{L}}

\newcommand{\GInter}[2]{#1 \rightarrow #2}

\newcommand{\ROLE}[1]{\mathtt{#1}}
\newcommand{\RecT}{\keyword{t}}

\newcommand{\Gmodular}{\ensuremath{\ell}}

\newcommand{\goto}{\textbf{call} \ }
\newcommand{\gotofun}[2]{#1 \Downarrow #2}
\newcommand{\modularise}[1]{\mathbb{L}(#1)}
\newcommand{\actset}[2]{ [  \! [ #1 ] \! ]_{#2}}

\newcommand{\DTRANS}[1]{\xrightarrow{#1}_{\scriptsize {D}}}
\newcommand{\DDTRANS}[1]{\stackrel{#1~~}{\Longrightarrow_{\scriptsize {D}}}}
\newcommand{\DSTRANS}[1]{\stackrel{~#1~~*}{\longrightarrow_{\scriptsize {D}}}}

\newcommand{\INFER}[2]{\frac{\displaystyle{#1}%
\vspace{2mm}%
}{
\vspace{2mm}%
\displaystyle{#2}
}}

{\end{array}%
 \right.}

\newcommand{\mrule}[1]{{\footnotesize{\ensuremath{[\text{\sc{#1}}]}}}}

\newcommand{\keyword}[1]{\textsf{\upshape #1}}
\newcommand{\kw}{\keyword} %simplified \keyword

\newcommand{\sas}{\mathcal{A}}
\newcommand{\sbs}{\mathcal{B}}
\newcommand{\scs}{\mathcal{C}}
\newcommand{\gotoS}{\textbf{call}}
\newcommand{\pa}[2]{(#1,#2)}
\title{Lightening Global Types}

\author{Tzu-chun Chen
\institute{Dipartimento di Informatica, Universit\`a di Torino, Italy}
\email{chen@di.unito.it}
}

\def \titlerunning{Lightening Global Types}
\def \authorrunning{Tzu-chun Chen}

\clearpage

\begin{document}
\maketitle

%------------
\begin{abstract}
Global session types prevent participants from waiting for never coming messages. 
Some interactions take place just for the purpose of informing receivers
that some message will never arrive or the session is terminated. By decomposing a big global type into several light global types, one can avoid such kind of redundant interactions. Lightening global types gives us cleaner global types,  which keep all necessary communications. This work proposes a framework which allows to easily decompose global types into light global types,  preserving the interaction sequences of the original ones but for redundant interactions.
\end{abstract}
%------------

\section{Introduction}
\label{sec:introduction}
Since 
cooperating tasks and sharing resources through communications 
under network infrastructures (e.g. clouds, large-scale distributed systems, etc.)
has become the norm and the services for communications are growing with increasing users,
it is a need to give programmers 
an easy and powerful programming language for developing applications of interactions.
For this aim, Scribble \cite{HondaMBCY11}, 
a communication-based programming language,
is introduced building on the theory of global types \cite{HYC08,BettiniCDLDY08}.
A developer can use Scribble to code a global protocol, which
stipulates any local endpoints (i.e. local applications) participanting in it.
The merits of coding global protocols, 
rather than just coding the local ones,  
are (1) giving all local participants a clear blue map of
what events they are involved in and what purposes of those events and
(2) making it easier and more efficient to exchange, share, and maintain communications plans 
(e.g. design of global protocols) across organisations.
However, the tool itself cannot ensure an efficient communication programming. 
The scenario of a global communication can be very complicated so
it becomes a burden for programmers 
to correctly code interactions which satisfy protocols (described by global types).
At runtime, 
the cost for keeping all resources ready for a long communication
and for maintaining the safety of the whole system can increase a lot. 

%%%%%%%%%% first issue (contribution) %%%%%%%%%%%%%%%%%%%%%
For example, assume a gift requester needs a \textit{key} (with her identity) to get a wanted gift.
In order to get the key, she needs to get a \textit{guide}, which is a map for finding the key.
It is like searching for treasures step by step, 
where a player needs not be always online in {\em one} session
for completing the whole procedure. 
Instead, the communication protocol
can be viewed as {\em separated but related} sessions
which are linked (we use $\gotoS$s to switch from a session to another).
For example, 
let a session $\sas$ do the interactions for getting \textit{guide},
and another session $\sbs$ do the interactions for getting \textit{key} with \textit{guide}. 
Both \textit{guide} and \textit{key} are knowledge gained from these interactions.
\textit{guide} bridges $\sas$ and $\sbs$ as it is gained in $\sas$ and used in $\sbs$.
The participant then uses \textit{key}, if she successfully got it in $\sbs$,
to gain the wanted gift. Let session $\scs$ implements these final interactions.

As standard we use global types \cite{HYC08,BettiniCDLDY08} to describe interaction protocols, adding a $\goto$ command and type declarations for relating sessions. We call {\em light global types} the global types written in the extended syntax.  
Figure \ref{fig:intro:sepsession} simply represents
the difference between viewing all interactions as one scenario and 
viewing interactions as three separated ones. As usual $\GInter{\ROLE{store}}{\ROLE{req}}: 
\{ \textit{\color{Green}{yes3}}(\stringk). \kend, \textit{\color{Green}{no3}}( ) . \kend \}$ models a communication where role $\ROLE{store}$ sends to role $\ROLE{req}$ either the label $\textit{\color{Green}{yes3}}$ and a value of type $\stringk$ or the label $\textit{\color{Green}{no3}}$, and in both cases $\kend$ finishes the interaction. The construct $\goto\Gmodular_b$ indicates that the interaction should continue by executing the session described by the 
light global type associated to the name $\Gmodular_b$.

\begin{figure}
{
$$
\begin{array}{ll}
\begin{array}{rcl}
\GAM & = &
\GInter{\ROLE{req}}{\ROLE{map}}: {\color{blue}\textit{req1}}(\stringk) . \\
& &
  \GInter{\ROLE{map}}{\ROLE{req}}: \\
  & & \{ \ {\color{blue}\textit{yes1}}(\stringk). \\
  & & \ \ \ 
     \GInter{\ROLE{req}}{\ROLE{issuer}}: {\color{red}\textit{req2}}(\stringk). \\
      & &  \ \ \ 
      \GInter{\ROLE{issuer}}{\ROLE{req}}:\\
      & & \ \ \ \{ \  {\color{red}\textit{yes2}}(\intk).\\
       & & \ \ \ \ \ \ 
             \GInter{\ROLE{req}}{\ROLE{store}} : {\color{Green}\textit{req3}}(\intk). \\
       & & \ \ \ \ \ \ 
             \GInter{\ROLE{store}}{\ROLE{req}}: \\
       & & \qquad \{ \ 
              { \textit{\color{Green}{yes3}}(\stringk) . \kend
              ,
              \textit{\color{Green}{no3}}()}. \kend  \ \}, \\
   & & \ \ \ \phantom{ \{ }{\color{red} \ \textit{no2}}() . 
                  \GInter{\ROLE{req}}{\ROLE{store}}: \textit{no4}().\kend \ \}, \\
& &  \phantom{\{ \ } {\color{blue} \textit{no1}}().\\
& &   \ \ \ \
         \GInter{\ROLE{req}}{\ROLE{issuer}}: \textit{no5}().\\
%& &   \ \ \   
%         \GInter{\ROLE{issuer}}{\ROLE{req}}: \textit{no6}().\\
& &   \ \ \ \
         \GInter{\ROLE{req}}{\ROLE{store}}: \textit{no6}().
\kend \ \} \\
\end{array}
&
\begin{array}{rcl}
\Gmodular_a & = &
\GInter{\ROLE{req}}{\ROLE{map}}: {\color{blue}\textit{req1}}(\stringk) . \\
& &
\GInter{\ROLE{map}}{\ROLE{req}}: \\
& &
\{ \ {\color{blue}\textit{yes1}}(\stringk). \goto \Gmodular_b 
,  {\color{blue}\textit{\textit{no1}}}( ) . \kend  \ \}  \\\\
\Gmodular_b & = &
\GInter{\ROLE{req}}{\ROLE{issuer}}: {\color{red}\textit{req2}}(\stringk) . \\
& &
\GInter{\ROLE{issuer}}{\ROLE{req}}: \\
& & 
\{ \ {\color{red} \textit{yes2}}(\intk).\goto \Gmodular_c 
, {\color{red}\textit{no2}}( ) . \kend  \ \} \\\\
\Gmodular_c & = &
\GInter{\ROLE{req}}{\ROLE{store}}: {\color{Green}\textit{req3}}(\intk) . \\
& &
\GInter{\ROLE{store}}{\ROLE{req}}: \\
& & 
\{ \ \textit{\color{Green}{yes3}}(\stringk).\kend
, \textit{\color{Green}{no3}}( ) . \kend \ \}
\end{array}
\end{array}
$$
\caption{Viewing interactions as a whole or as separated but related ones.
\label{fig:intro:sepsession}}
}
\end{figure}

The set of light global types associated to the names $\Gmodular_a$, $\Gmodular_b$ and  $\Gmodular_c$ describe the same 
protocol given by the global type $\GAM$. An advantage of lightening is to spare communications needed only to warn participants they will not receive further messages. 
We call such communications {\em redundant} ones. In the example we avoid
the interactions $\GInter{\ROLE{req}}{\ROLE{store}}: \textit{no4}()$ and
$\GInter{\ROLE{req}}{\ROLE{issuer}}: \textit{no5}().\GInter{\ROLE{req}}{\ROLE{store}}: \textit{no6}()$.
Moreover, lightening prevents both local participants and global network  
from wasting resources, e.g. keeping online or waiting for step-by-step permissions.

\begin{figure}
\[\small
\begin{array}{ll}
\begin{array}{l}
\protocolS\ \text{getGift}\\[0.03mm]
\qquad \qquad 
(\roleS\ \ROLE{req}, 
\roleS\ \ROLE{map},\\[0.03mm]
\qquad \qquad \
\roleS\ \ROLE{issuer},
\roleS\ \ROLE{store})\{\\[0.03mm]
\quad
{\color{blue}\textit{req1}}(\stringk\ \text{identity})\
\fromS\ \ROLE{req}\ \toS\ \ROLE{map}; \\[0.03mm]

\quad 
\choiceS\ \atS\ \ROLE{map}\{
\\[0.03mm]
\qquad 
{\color{blue}\textit{yes1}}(\stringk\ \text{guide})\
\fromS\ \ROLE{map}\ \toS\ \ROLE{req}; \\[0.03mm]
\qquad 
{\color{red}\textit{req2}}(\stringk\ \text{guide})\
\fromS\ \ROLE{req}\ \toS\ \ROLE{issuer}; \\[0.03mm]
\qquad 
\choiceS\ \atS\ \ROLE{issuer} \{ \\[0.03mm]
\qquad \quad
   {\color{red}\textit{yes2}}(\intk\ \text{key})\
   \fromS\ \ROLE{issuer}\ \toS\ \ROLE{req}; \\[0.03mm]
\qquad \quad
   {\color{Green}\textit{req3}}(\intk\ \text{key})\
   \fromS\ \ROLE{req}\ \toS\ \ROLE{store}; \\[0.03mm]
\qquad \quad      
\choiceS\ \atS\ \ROLE{store} \{ \\[0.03mm]
\qquad \quad \quad
  {\color{Green}\textit{yes3}}(\stringk\ \text{gift})\
   \fromS\ \ROLE{store}\ \toS\ \ROLE{req}; \\[0.03mm]

\qquad \quad
\} \ \orS\ \{ \\[0.03mm]
\qquad \quad \quad 
 {\color{Green}\textit{no3}}()\
   \fromS\ \ROLE{store}\ \toS\ \ROLE{req}; \\[0.03mm]
\qquad \quad
\} \\[0.03mm]
\qquad
\} \ \orS\ \{ \\[0.03mm]
\qquad \quad
  {\color{red}\textit{no2}}()\
  \fromS\ \ROLE{issuer}\ \toS\ \ROLE{req};\\[0.03mm]
\qquad \quad
  \textit{no4}()\
  \fromS\ \ROLE{req}\ \toS\ \ROLE{store};\\[0.03mm]
\qquad 
\}
\\[0.03mm]
\quad 
\}\ \orS\ \{ \\[0.03mm]
\qquad 
{\color{blue}\textit{no1}()}\
\fromS\ \ROLE{map}\ \toS\ \ROLE{req} ;\\[0.03mm]
\qquad 
{\textit{no5}}()\
\fromS\ \ROLE{req}\ \toS\ \ROLE{issuer}; \\[0.03mm]
\qquad
{\textit{no6}}()\
\fromS\ \ROLE{req}\ \toS\ \ROLE{store}; \\[0.03mm]
\quad\}
\\
\}
\end{array}
& 
\begin{array}{l}
\protocolS\ \text{getGuide}\ 
(\roleS\ \ROLE{req}, 
\roleS\ \ROLE{map})\{\\[0.03mm]
\quad
{\color{blue}\textit{req1}}(\stringk\ \text{identity})\
\fromS\ \ROLE{req}\ \toS\ \ROLE{map}; \\[0.03mm]
\quad
\choiceS\ \atS\ \ROLE{map}\{\\[0.03mm]
\qquad \
{\color{blue}\textit{yes1}}(\stringk\ \text{guide})\
\fromS\ \ROLE{map}\ \toS\ \ROLE{req}; \\[0.03mm] 
\qquad \
\runS\ \protocolS\ \\[0.03mm]
\qquad \qquad \quad
\text{getKey}(\roleS\ \ROLE{req}, \roleS\ \ROLE{issuer})\ \atS\ \ROLE{req}; \\[0.03mm]
\quad 
\} \ \orS\ \{ \\[0.03mm]
\qquad \
{\color{blue}\textit{no1}}()\
\fromS\ \ROLE{map}\ \toS\ \ROLE{req}; \\[0.03mm]
\quad 
\} 
\\[0.03mm]
\}
\\[0.1mm]
\protocolS\ \text{getKey}\ 
(\roleS\ \ROLE{req}, 
\roleS\ \ROLE{issuer} )\{\\[0.03mm]
\quad
{\color{red}\textit{req2}}(\stringk\ \text{guide})\
\fromS\ \ROLE{req}\ \toS\ \ROLE{issuer}; \\[0.03mm]
\quad
\choiceS\ \atS\ \ROLE{issuer}\{\\[0.03mm]
\qquad 
{\color{red}\textit{yes2}}(\intk\ \text{key})\
\fromS\ \ROLE{issuer}\ \toS\ \ROLE{req}; \\[0.03mm]
\qquad
\runS\ \protocolS\ \\[0.03mm]
\qquad \qquad \quad
\text{getGift}'(\roleS\ \ROLE{req}, \roleS\ \ROLE{store})\ \atS\ \ROLE{req}; \\[0.03mm]
\quad 
\} \ \orS\ \{ \\[0.03mm]
\qquad 
{\color{red}\textit{no2}}()\
\fromS\ \ROLE{issuer}\ \toS\ \ROLE{req}; \\[0.03mm]
\quad 
\} 
\\[0.03mm]
\}
\\[0.1mm]
\protocolS\ \text{getGift}'\ 
(\roleS\ \ROLE{req}, 
\roleS\ \ROLE{store} )\{\\[0.03mm]
\quad
{\color{Green}\textit{req3}}(\intk\ \text{key})\
\fromS\ \ROLE{req}\ \toS\ \ROLE{store}; \\[0.03mm]
\quad
\choiceS\ \atS\ \ROLE{store}\{\\[0.03mm]
\qquad
{\color{Green}\textit{yes3}}(\stringk\ \text{gift})\
\fromS\ \ROLE{store}\ \toS\ \ROLE{req}; \\[0.03mm] 
\quad 
\} \ \orS\ \{ \\[0.03mm]
\qquad
{\color{Green}\textit{no3}}()\
\fromS\ \ROLE{store}\ \toS\ \ROLE{req}; \\[0.03mm]
\quad 
\} 
\\[0.03mm]
\}
\end{array}
\end{array}
\]
\caption{Scribble 
for $\protocolS$s getGift (LHS) and 
for getGuide, getKey, and $\text{getGift}'$ (RHS).
\label{fig:intro:scribble}}
\end{figure}

Features of lightening become more clear by looking at the Scribble code implementing the global types of Figure~\ref{fig:intro:sepsession}, see Figure~\ref{fig:intro:scribble}. The words in bold are keywords.
In Scribble  after the  \protocolS\ keyword   one writes
the protocol's name
and declares the roles involved in the session,
then one describes the interactions in the body.
The left-hand side (LHS) of Figure \ref{fig:intro:scribble}
shows the Scribble code corresponding to the global type $\GAM$.
Although the code is for a simple task,
it is involved
%involved 
since there are three nested $\choiceS$s.
On the contrary, the right-hand side (RHS) of Figure \ref{fig:intro:scribble}
clearly illustrates the steps for getting a gift in three small protocols (corresponding to 
the light global types associated to $\Gmodular_a$, $\Gmodular_b$, $\Gmodular_c$ respectively), 
which are separated but linked (by $\runS$ and $\atS$) for preserving the causality.

\bigskip

The structure of the paper is the following. 
Section \ref{sec:grammar} introduces 
a framework of light global session types,
which gives a syntax to easily compose a global type by light global types.
Section \ref{sec:decompose} proposes a function for 
decomposing a general global type into light global types, 
e.g. it decomposes $\GAM$ into the types associated to $\Gmodular_a, \Gmodular_b$ and $\Gmodular_c$.
Section \ref{sec:theory} proves the soundness of the function,
and Section \ref{sec:final} discusses related and future works. 

\section{Syntax of (light) global types}
\label{sec:grammar}
The syntax for {\em global types} is standard:
$$
\begin{array}{rclrcl}
\GAM & :: = & 
\GInter{\ROLE{r}_1}{\ROLE{r}_2}: \{l_j(\sort_j) . 
\GAM_j \}_{j \in J} 
\PAR
\mu \RecT . \GAM
\PAR
\RecT
\PAR
\kend&\hfill\text{where }\sort &::=& \unitk \PAR \natk \PAR \stringk \PAR \boolk 
\end{array}
$$
The branching 
$\GInter{\ROLE{r}_1}{\ROLE{r}_2}: \{ l_j(\sort_j) . \GAM_j \}_{j \in J}$ 
says that role $\ROLE{r}_1$ sends a label $l_j$ and a message of type $\sort_j$ to $\ROLE{r}_2$ by selecting $j\in J$ 
and then the interaction continues as described in $\GAM_j$.
$\mu \RecT . \GAM$ is a recursive type,
where $\RecT$ is guarded in $\GAM$ in the standard way.
$\kend$ means termination of the protocol. We write $l ()$ as short for $l(\unitk)$ and we omit brackets when there is only one branch.

{\em Light global types} are global types extended with declarations and the construct $\gotoS$:
$$
\begin{array}{rclrcl}
D &::=& \overrightarrow{\Gmodular= \M} %\PAR D~D
&\qquad\qquad
\M & :: = & 
\GInter{\ROLE{r}_1}{\ROLE{r}_2}: \{l_j(\sort_j) . 
\M_j \}_{j \in J} 
\PAR
\mu \RecT . \M
\PAR
\GAM
\PAR
\goto  \Gmodular
\end{array}
$$
{\em Declarations} associate names to light global types, for example in Figure \ref{fig:intro:sepsession} the name $\Gmodular_a $ is associated with the  type 
$$\begin{array}{l}
\GInter{\ROLE{req}}{\ROLE{map}}: {\color{blue}\textit{req1}}(\stringk) . ~~
\GInter{\ROLE{map}}{\ROLE{req}}: 
\{ \ {\color{blue}\textit{yes1}}(\stringk). \goto \Gmodular_b 
,  {\color{blue}\textit{\textit{no1}}}( ) . \kend  \ \}\end{array}$$
We denote by $\emptyset$ the empty declaration. 

The type $\goto  \Gmodular$ prescribes that the interaction continues by opening a new session with light global type $\M$ such that $\Gmodular=\M$ belongs to the set of current declarations. 
For example according to the declarations in Figure \ref{fig:intro:sepsession} $\goto  \Gmodular_b$ asks to open a new session with type:
$$\begin{array}{l}
\GInter{\ROLE{req}}{\ROLE{issuer}}: {\color{red}\textit{req2}}(\stringk) . ~~\GInter{\ROLE{issuer}}{\ROLE{req}}: \{ \ {\color{red} \textit{yes2}}(\intk).\goto \Gmodular_c 
, {\color{red}\textit{no2}}( ) . \kend  \ \}
\end{array}$$

\section{Lightening global types}
\label{sec:decompose}
This section describes a function for removing redundant interactions (Definition \ref{def:redundant})  from (possibly light) global types.
It uses lightening, since it adds $\gotoS$ constructors and declarations. 
In the next section we will show that the initial and final protocols describe the same non-redundant interactions. 

We consider an interaction redundant when only one label can be sent and the message is not meaningful, i.e. it has type $\unitk$ and it does not appear under a recursion. More precisely, by defining light global contexts (without recursion) as follows:
$$
\begin{array}{rcl}
\Context & :: = &  [ \;  ] \PAR\GInter{\ROLE{r}_1}{\ROLE{r}_2} : 
\{ l_j (\sort_j) . \GAM_j, l (\sort) . \Context \}_{j \in J} 
\end{array}
$$
we get:
\begin{DEF}[Redundant interaction] \label{def:redundant} \rm
The interaction 
$\GInter{\ROLE{r}_1}{\ROLE{r}_2}: l() $ is {\em redundant} in 
$$\Context[\GInter{\ROLE{r}_1}{\ROLE{r}_2}: \  l( ) . \M \ ].$$
\end{DEF}
Interactions sending only one label and with messages of type $\unitk$ are needed inside recursions when they terminate the cycle. For example the interaction $\GInter{\ROLE{r_2}}{\ROLE{r_3}}:\textit{stop}()$ cannot be erased in the type
$$\mu \RecT .\GInter{\ROLE{r_1}}{\ROLE{r_2}}:\{\textit{goon}(\intk).\GInter{\ROLE{r_2}}{\ROLE{r_3}}:\textit{goon}(\intk).\RecT,
\textit{stop}().\GInter{\ROLE{r_2}}{\ROLE{r_3}}:\textit{stop}().\kend\}$$

We define $\modularise{\M}$ as a function for 
removing redundant interactions in $\M$ by decomposing $\M$ into separated light global types. The basic idea is that, in order to erase redundant communications, the roles which are the receivers of these communications must  get the communications belonging to other branches in separated types. 
Therefore the result of $\modularise{\M}$ 
is a new light global type and a set of declarations with fresh names.

The function $\modularise{\M}$ uses the auxiliary function $\gotofun{\M}{\ROLE{r}}$ which searches inside the branches of $\M$ the first communications with receiver ${\ROLE{r}}$ and replaces these communications (and all the following types) by $\gotoS$s to newly created names, which are associated by declarations to the corresponding types. Therefore also the result of $\gotofun{\M}{\ROLE{r}}$ is a new light global type and a set of declarations with fresh names.

\begin{DEF}[The function $\Downarrow$]\label{dw}\rm
The function $\gotofun{\M}{\ROLE{r}}$ is defined by induction on $\M$:
$$
\begin{array}{c}
\gotofun{\kend}{\ROLE{r}} = \pa\kend\emptyset \qquad
\gotofun{\RecT}{\ROLE{r}} = \pa\RecT\emptyset \qquad
\gotofun{ \goto  \Gmodular}{\ROLE{r}} = \pa{\goto  \Gmodular}\emptyset\qquad
\gotofun{\mu \RecT . \M }{\ROLE{r}} =  \pa{\mu \RecT . \M}\emptyset\\\\
\gotofun{
(\GInter{\ROLE{r}_1}{\ROLE{r}_2}: \{ l_j (\sort_j) . \M_j \}_{j \in J})
}
{\ROLE{r}} =
\begin{cases}
\pa{\M_j}{\emptyset}& \text{if} \ \ROLE{r}_2 = \ROLE{r}\text{ and}\\ &\sort_j=\unitk \text{ and}\\& J=\{j\}\\
\pa{\goto \Gmodular}
{\Gmodular=
\GInter{\ROLE{r}_1}{\ROLE{r}_2} : \{ l_j (\sort_j) . \M_j \}_{j \in J}
}
& \text{if} \ \ROLE{r}_2 = \ROLE{r}\text{ and}\\ 
\qquad\qquad\text{where } \Gmodular \text{ is a  fresh name}&\sort_j\not=\unitk \text{ or}\\
& J\not=\{j\}\\
\pa{ \GInter{\ROLE{r}_1}{\ROLE{r}_2} : \{ l_j (\sort_j) . \M'_j \}_{j \in J}}
{\bigcup_{j\in J} D_j}
& \text{otherwise}\\
\hfill\text{where }\gotofun{\M_j}{\ROLE{r}} = \pa{\M'_j}{D_j} \text{ for } j\in J.
\end{cases} 
\end{array}
$$
\end{DEF}
We remark that the first case of the definition of $\Downarrow$ for a branching type allows this function to eliminate a redundant interaction. Therefore one application of function $\mathbb{L}$ can get rid of more than one redundant interaction, as exemplified below. 

The function $\mathbb{L}$ is defined by induction on the context in which the redundant interaction appears. 

\begin{DEF}[The function $\mathbb{L}$] \label{def:modularise} \rm  The application of the function  $\mathbb{L}$ to $\Context[\GInter{\ROLE{r}_1}{\ROLE{r}_2}: \  l() . \M\   ]$ for eliminating the interaction $\GInter{\ROLE{r}_1}{\ROLE{r}_2}:  \ l() . \M\ $ is defined by induction on $\Context$:

$$
\begin{array}{lcl}
\modularise{\GInter{\ROLE{r}_1}{\ROLE{r}_2}:  \ l(  ) . \M\ }&=&\pa \M\emptyset
\end{array}$$
$$
\begin{array}{lcl}
\modularise{\GInter{\ROLE{r}'_1}{\ROLE{r}'_2} : 
\{ l_j (\sort_j) .\M_j, 
    l '(\sort) . \Context[\GInter{\ROLE{r}_1}{\ROLE{r}_2}:  \ l(  ) . \M\ ]\}_{j \in J}}   = \\\\
\begin{cases}
%\begin{array}{l}
\pa{\GInter{\ROLE{r}'_1}{\ROLE{r}'_2} : 
\{ l_j (\sort_j) . \M_j,
   l' (\sort) . \Context[ \M  ] \}_{j \in J}}\emptyset
%\end{array}
&
\text{if} \ \ROLE{r}_2 = \ROLE{r}'_1\ 
\text{or} \ \ROLE{r}_2 = \ROLE{r}'_2,\\
%\begin{array}{l}
\pa{\GInter{\ROLE{r}'_1}{\ROLE{r}'_2} : 
\{ l_j (\sort_j) . 
  \M_j', 
   l' (\sort) . \M' \}_{j \in J}}{\bigcup_{j\in J} D_j\cup D}
%\end{array}
& %\! \! \!
%\begin{array}{l}
\text{if }  \ROLE{r}_2  \not = \ROLE{r}'_1\ 
\text{ and } \ \ROLE{r}_2  \not = \ROLE{r}'_2,\\
\qquad\qquad\qquad\qquad\text{where }
\gotofun{\M_j}{\ROLE{r}_2} =\pa{\M'_j}{D_j} \text{ for } j\in J\\
\qquad\qquad\qquad\qquad\qquad \text{and }\modularise{\Context[\GInter{\ROLE{r}_1}{\ROLE{r}_2}: \  \ l() . \M\   ]} =\pa{\M'}{D}
%\\
%$\;$ 
%\end{array}
\end{cases}
\end{array}
$$
\end{DEF}

The branching case of the above definitions needs some comments. If the receiver $\ROLE{r}_2$ is also the sender or the receiver of the top branching, then she is aware of the choice of the label $l'$ and the redundant interaction can simply be erased. Otherwise $\ROLE{r}_2$  must receive a communication in all branches $l_j$ and in this case these communications need to be replaced by $\gotoS$s to fresh names of light global types. For this reason we compute $\gotofun{\M_j}{\ROLE{r}_2}$ for all $j\in J$. Moreover we recursively call the mapping $\mathbb{L}$ on $\Context[\GInter{\ROLE{r}_1}{\ROLE{r}_2}: \  \ l() . \M \   ]$.

We now show two applications of the lightening function, first to the global type $\GAM$ of Figure \ref{fig:intro:sepsession} and then to the light global type obtained as a result of the first application. 
The application of $\mathbb{L}$ for eliminating the interaction $\GInter{\ROLE{req}}{\ROLE{store}}: \textit{no4}(  )$ gives as a result $\pa{\M'_a} {\Gmodular_c = \M_c}$, where $\M_c$ is the type associated to $\Gmodular_c$ in Figure \ref{fig:intro:sepsession} and:
$$\begin{array}{rcl}
\M'_a &=&
\GInter{\ROLE{req}}{\ROLE{map}}: {\color{blue}\textit{req1}}(\stringk).
  \GInter{\ROLE{map}}{\ROLE{req}}: \\
  & & \{ \ {\color{blue}\textit{yes1}}(\stringk). 
            \GInter{\ROLE{req}}{\ROLE{issuer}}: {\color{red}\textit{req2}}(\stringk). 
            \GInter{\ROLE{issuer}}{\ROLE{req}}:\\
      & & \ \ \ \{ \  {\color{red}
             \textit{yes2}}(\intk). \goto \Gmodular_c,
            {\color{red} 
            \textit{no2}}(  ). \kend \}, \\
& & \ \ {\color{blue} \textit{no1}}( ).
         \GInter{\ROLE{req}}{\ROLE{issuer}}: \textit{no5}(  ).\kend  \} 
\end{array}$$
Notice that also the redundant interaction $\GInter{\ROLE{req}}{\ROLE{store}}: \textit{no6}(  )$ is erased in $\M'_a$. 
The same result can be obtained by applying $\mathbb{L}$ for eliminating the interaction $\GInter{\ROLE{req}}{\ROLE{store}}: \textit{no6}(  )$. 
Now if we apply $\mathbb{L}$ to $\M'_a$ for eliminating $ \GInter{\ROLE{req}}{\ROLE{issuer}}: \textit{no5}(  )$ 
we get $\pa{\M_a}{\Gmodular_b = \M_b}$, \MD{where $\M_a$ and $\M_b$ are the types associated to 
$\Gmodular_a$ and $\Gmodular_b$ in Figure \ref{fig:intro:sepsession}. So we get all declarations shown in Figure~\ref{fig:intro:sepsession}.}
% and $\pa{\M_b}{\Gmodular_c = \M_c}$, 
%\TZ{where $\M_b$ and $\M_c$ are the types associated to 
%$\Gmodular_b$ and $\Gmodular_c$ in Figure \ref{fig:intro:sepsession}}.

\section{Safety of the  lightening function}
\label{sec:theory}
In order to discuss the correctness of our lightening function, following \cite{CDP12} we view light global types with relative sets of declarations as denoting languages of interactions which can occur in multi-party sessions. The only difference is that recursive types do not reduce, the reasons being that our lightening function does not modify them. More formally the following definition gives a labelled transition system for light global types with respect to a fixed set of declarations. As usual $\tau$ means a silent move, and $\checkmark$ session termination. 
%The labelled transition rules of $\DTRANS{\lambda}$ are defined below.
\medskip

\begin{DEF}[LTS] \label{def:actionrule} \rm
$\;$\\[10pt]
$$
\begin{array}{c}
\begin{array}{l}
\mrule{call}
\\[0.3mm]
\INFER{\Gmodular = \M \in D }
{\goto \Gmodular \DTRANS{\tau} \M} 
\end{array}
\begin{array}{l}
\mrule{rec}
\\[0.3mm]
\mu \RecT . \M \DTRANS{\checkmark} 
\end{array}
\begin{array}{l}
\mrule{red}
\\[0.3mm]
\GInter{\ROLE{r}_1}{\ROLE{r}_2}:\ l ( ) . \M\ \DTRANS{\tau} \M
\end{array}
\begin{array}{l}
\mrule{end}
\\[0.3mm]
\kend \DTRANS{\checkmark}
\end{array}\\[20pt]
\begin{array}{l}
\mrule{act}
\\[0.5mm]
\INFER{  J \not=\{j\} \ \text{or} \ \sort_j \not =\unitk 
}{
\GInter{\ROLE{r}_1}{\ROLE{r}_2} : \{ l_j (\sort_j) . \M_j \}_{j \in J}
\DTRANS{\GInter{\ROLE{r}_1}{\ROLE{r}_2} :  l_j (\sort_j) } \M_j
}
\end{array}
\end{array} 
$$
\end{DEF}

We convene that $\lambda$ ranges over $\GInter{\ROLE{r}_1}{\ROLE{r}_2} :  l (\sort)$ and $\checkmark$, and that $\sigma$ ranges over sequences of $\lambda$. Using the standard notation:
$$\begin{array}{rcl}
\M \DDTRANS{\lambda} \M' \ &\text{if} & \ 
\M \DSTRANS{\tau} \M_1 \DTRANS{\lambda} \M_2 \DSTRANS{\tau} \M'\text{ for some } \M_1, \M_2\\
\M \DDTRANS{\lambda \cdot \sigma} \M' \ &\text{if} & \
\M \DDTRANS{\lambda} \M_1 \DDTRANS{\sigma} \M'\text{ for some } \M_1\\
\M \DDTRANS{\sigma\cdot\checkmark} \ &\text{if} & \
\M \DDTRANS{\sigma} \M_1 \DDTRANS{\checkmark} \text{ for some } \M_1
\end{array}
$$

The language generated by $\M$ and relative to $D$ is the set of sequences $\sigma$ obtained by reducing $\M$ using $D$. We take this language as the meaning of $\M$ relative to $D$ (notation $\actset{\M}{D}$).
%Define $\actset{M}{D}$ as a collection of actions taken by $M$.

\medskip

\begin{DEF}[Semantics] \label{def:actset} \rm
$\actset{\M}{D} = 
\{ \sigma \PAR
    \M \DDTRANS{\sigma} \M' \ \text{for some}\ \M' \ \text{or}\ \M \DDTRANS{\sigma}
\}
$.
\end{DEF}
Soundness of lightening then amounts to show that the function $\mathbb{L}$ preserves the meaning of light global types with respect to the relative sets of declarations. A first lemma shows soundness of the function $\Downarrow$. 

\medskip

\begin{LEM} \label{len:gotofun} \rm
Let $\M$ be a light global type with declaration $D$.
If $\gotofun{\M}{\ROLE{r}} = (\M', D')$,
then $\actset{\M}{D} = \actset{\M'}{D' \cup D}$.
\end{LEM}
\begin{proof}
The proof is by induction on $\M$ and by cases on Definition \ref{dw}. The only interesting case is $\M = \GInter{\ROLE{r}_1}{\ROLE{r}_2}: \{ l_j (\sort_j). \M_j \}_{j \in J}$.
\begin{enumerate}
\item If $\gotofun{\M}{\ROLE{r}} = 
(\M_j, \emptyset)$,
then $\sort_j = \unitk$ and $J = \ASET{j}$.
By rule \mrule{red} we have
$\M \DTRANS{\tau} \M_j$,
thus $\actset{\M}{D} = \actset{\M_j}{D} = \actset{\M'}{D}$.

\item If $\gotofun{\M}{\ROLE{r}} =
(\goto \Gmodular, \Gmodular = \GInter{\ROLE{r}_1}{\ROLE{r}_2}: \{ l_j (\sort_j). \M_j \}_{j \in J})$,
then by rule \mrule{call} we have 
\TZ{$\M' \xrightarrow{\tau}_{\ASET{\Gmodular=\M} \cup D} \M$}.
Thus $\actset{\M}{D} = \actset{\M'}{\TZ{\ASET{\Gmodular=\M} \cup D}}$.

\item If $\gotofun{\M}{\ROLE{r}}=
(\GInter{\ROLE{r}_1}{\ROLE{r}_2}: \{ l_j (\sort_j). \gotofun{\M'_j}{\ROLE{r}} \}_{j \in J}, \bigcup_{j \in J} D_j)$,
then $\gotofun{\M_j}{\ROLE{r}} = ( \M'_j, D_j)$ for $j \in J$.
By Definition \ref{def:actset} and rule \mrule{act} $\actset{\M}{D} =\bigcup_{j \in J} \{ \ \GInter{\ROLE{r}_1}{\ROLE{r}_2} : l_j (\sort_j) \cdot \sigma \mid \sigma \in \actset{\M_j}{D}\ \}$ and
$\actset{\M'}{\bigcup_{j \in J} D_j\cup D} =\bigcup_{j \in J} \{ \ \GInter{\ROLE{r}_1}{\ROLE{r}_2} : l_j (\sort_j) \cdot \sigma \mid \sigma \in \actset{\M'_j}{\bigcup_{j \in J} D_j\cup D}\ \}$. 
Since by induction $\actset{\M_j}{D} = \actset{\M'_j}{D_j \cup D}$, we conclude $\actset{\M}{D} = \actset{\M'}{D' \cup D}$.
\end{enumerate}
\end{proof}

\begin{THM}[Soundness] \label{thm:safem}\rm 
Let $\M$ be a light global type with set of declarations $D$.
If $\modularise{\M} = (\M', D')$,
then $\actset{\M}{D} = \actset{\M'}{D' \cup D}$.
\end{THM}

\begin{proof}
The proof is by induction on $\M$ and by cases on Definition \ref{def:modularise}. 
\begin{enumerate}
\item If $\M = \GInter{\ROLE{r}_1}{\ROLE{r}_2}:  l(  ) . \M''$, then
         $\modularise{\M} = (\M'', \emptyset)$. We conclude since by Definition \ref{def:actset} and 
   by rule \mrule{red}
                 $\actset{\M}{D} = \actset{\M''}{D}$.
         
\item Let $\M = \Context_0[\GInter{\ROLE{r}_1}{\ROLE{r}_2}:  l(  ) . \M''  ]$ where
         $\Context_0 = 
            \GInter{\ROLE{r}'_1}{\ROLE{r}'_2} : 
                \{ l_j (\sort_j) .\M_j, 
                 l '(\sort) .\Context \}]\}_{j \in J}$. By Definition \ref{def:actset} and 
   by rule \mrule{act} 
   $$\begin{array}{lll}\actset{\M}{D}& =& \bigcup_{j \in J} \{ \ \GInter{\ROLE{r'}_1}{\ROLE{r'}_2} : l_j (\sort_j) \cdot \sigma \mid \sigma \in \actset{\M_j}{D}\ \}\cup\\
   &&  \{ \ \GInter{\ROLE{r'}_1}{\ROLE{r'}_2} : l' (\sort) \cdot \sigma \mid \sigma \in \actset{\Context[\GInter{\ROLE{r}_1}{\ROLE{r}_2}:  l(  ) . \M''  ]}{D}\ \}\end{array}$$
                 \begin{enumerate}\item If
     $\ROLE{r}_2 = \ROLE{r}'_1$ or $\ROLE{r}_2 = \ROLE{r}'_2$, then 
     $$\begin{array}{lll}\actset{\M'}{D}& =& \bigcup_{j \in J} \{ \ \GInter{\ROLE{r'}_1}{\ROLE{r'}_2} : l_j (\sort_j) \cdot \sigma \mid \sigma \in \actset{\M_j}{D}\ \}\cup\\
   &&  \{ \ \GInter{\ROLE{r'}_1}{\ROLE{r'}_2} : l' (\sort) \cdot \sigma \mid \sigma \in \actset{\Context[ \M''  ]}{D}\ \}\end{array}$$ 
   We conclude since rule \mrule{red} $\actset{\Context[\GInter{\ROLE{r}_1}{\ROLE{r}_2}:  l(  ) . \M''  ]}{D}=\actset{\Context[ \M''  ]}{D}$.

\item If $\ROLE{r}_2 \not = \ROLE{r}'_1$ and $\ROLE{r}_2 \not = \ROLE{r}'_2$, let us    
assume    $\gotofun{\M_j}{\ROLE{r}_2} =\pa{\M'_j}{D_j}$ for $ j\in J$ and\\ $\modularise{\Context[\GInter{\ROLE{r}_1}{\ROLE{r}_2}: \  \ l() . \M\   ]}=\pa{\M_0'}{D_0}$. Then\\[4pt]
     $\begin{array}{lll}\actset{\M'}{\bigcup_{j\in J} D_j\cup D_0\cup D} &=& \bigcup_{j \in J} \{ \ \GInter{\ROLE{r'}_1}{\ROLE{r'}_2} : l_j (\sort_j) \cdot \sigma \mid \sigma \in \actset{\M'_j}{\bigcup_{j\in J} D_j\cup D_0\cup D}\ \}\cup\\
   &&  \{ \ \GInter{\ROLE{r'}_1}{\ROLE{r'}_2} : l' (\sort) \cdot \sigma \mid \sigma \in \actset{\M'_0}{\bigcup_{j\in J} D_j\cup D_0\cup D}\
    \}\end{array}$
   
   \noindent
   We conclude since by Lemma \ref{len:gotofun} $\actset{\M_j}{D}=\actset{\M'_j}{D_j\cup D}$ and by induction 
   $$\actset{\Context[\GInter{\ROLE{r}_1}{\ROLE{r}_2}:  l(  ) . \M''  ]}{D}=\actset{ \M'_0  }{D_0\cup D}.$$
                   \end{enumerate}
\end{enumerate}
In the running example we get $\actset{G}{\emptyset}=\actset{\M'_a}{\Gmodular_c = \M_c}=\actset{\M_a}{\Gmodular_b = \M_b~\Gmodular_c = \M_c}$.
\end{proof}

\section{Related works and conclusion}
\label{sec:final}
In the recent literature global types have been enriched in various directions by making them more expressive through roles~\cite{DY11} or  logical assertions~\cite{BHTY10} or monitoring~\cite{BocchiCDHY13}, more safe through 
 security levels for data and participants~\cite{CCDR10,ccd11} or reputation systems~\cite{bonotgc11}.

The more related papers are~\cite{DemangeonH12} and~\cite{CM12}.  Demangeon and Honda introduce nesting of protocols, that is, the possibility to define a subprotocol independently of its parent protocol, which calls the subprotocol explicitly. Through a call, arguments can be passed, such as values, roles and other protocols, allowing higher-order description. Therefore global types in~\cite{DemangeonH12} are much more expressive than our light types. Carbone and Montesi~\cite{CM12} propose to merge  together protocols interleaved in the same choreography into a single global type, removing costly invitations. Their approach is  opposite to ours, and they deal with implementations, while we deal with types. 

In this paper we show how 
to decompose interactions among multiple participants
in order to remove redundant interactions, by preserving the meaning of (light) global types.
We plan to implement our lightening function in order to experiment its practical utility in different scenarios. 

\vspace*{-10pt}

\paragraph*{\textbf{Acknowledgements}}
We are grateful to Mariangiola Dezani-Ciancaglini for her valuable comments and feedbacks.
We also thank PLACES reviewers for careful reading, 
since we deeply revised this article following their suggestions.

\nocite{*}
\bibliographystyle{eptcs}
\bibliography{session}

\end{document}